\documentclass[12pt]{amsart}
\usepackage[a4paper,margin=2.5cm]{geometry}
\usepackage[T1]{fontenc}
\usepackage[english]{babel}
\usepackage{microtype}
\usepackage{tikz}
\usepackage{tikz-cd}
\usepackage{amsmath,amssymb,amsthm,amscd,mathtools}
\usepackage{mathrsfs}
\usepackage{enumitem}
\usepackage{xcolor}
\usepackage{xcolor}
\definecolor{citeblue}{RGB}{0,70,140}
\usepackage[
  colorlinks=true,
  linkcolor=citeblue,
  citecolor=citeblue,
  urlcolor=citeblue
]{hyperref}
\usepackage{etoolbox}

\makeatletter
\patchcmd{\@setaddresses}
  {\par\addvspace\bigskipamount\indent}
  {\par\addvspace\bigskipamount\noindent}
  {}{}
\makeatother
\makeatletter
\def\l@subsection{\@tocline{2}{0pt}{3pc}{6pc}{}}
\makeatother
\makeatletter
\renewcommand\section{\@startsection{section}{1}%
  \z@{.7\linespacing\@plus\linespacing}{.5\linespacing}%
  {\normalfont\bfseries\centering}}

\renewcommand\subsection{\@startsection{subsection}{2}%
  \z@{.5\linespacing\@plus.7\linespacing}{-.5em}%
  {\normalfont\bfseries}}
\makeatother

\numberwithin{equation}{section}

\newtheorem{theorem}{Theorem}[section]
\newtheorem{prop}[theorem]{Proposition}
\newtheorem{lemma}[theorem]{Lemma}
\newtheorem{corollary}[theorem]{Corollary}
\newtheorem{definition}[theorem]{Definition}
\newtheorem*{remark}{Remark}

\DeclareMathOperator{\Det}{det}

\newcommand{\CS}{\operatorname{CS}}
\newcommand{\Tor}{\operatorname{Tor}}

\newcommand{\rk}{\operatorname{rk}}

\title[Functional--Integral Construction of Toral Chern--Simons Theory]{A Rigorous Functional--Integral Construction of Toral Chern--Simons Theory}
\author{Daniel Galviz}
\address{\noindent  YAU MATHEMATICAL SCIENCES CENTER AND DEPARTMENT OF MATHEMATICS, TSINGHUA
UNIVERSITY, BEIJING, CHINA.}
\date{}

\begin{document}

\begin{abstract}
We construct the functional integral of Abelian Chern--Simons theory with toral gauge group $\mathbb T=\mathfrak t/\Lambda \cong U(1)^n$ at level $K$, where $K:\Lambda\times\Lambda\to\mathbb Z$ is an even, integral, nondegenerate symmetric bilinear form, by exact zeta-regularized Gaussian evaluation of the formal quotient integral over connections modulo gauge. For closed $3$-manifolds, this yields a topological invariant; for manifolds with boundary, the relative functional integral produces the canonical boundary state. The resulting theory satisfies the required axioms of a $(2+1)$-dimensional TQFT.
\end{abstract}
\maketitle
{\scriptsize
\setlength{\parskip}{0pt}
\hypersetup{linkcolor=black}
\tableofcontents
}
\section{Introduction}
Chern--Simons theory admits several complementary mathematical descriptions.
Following Witten \cite{Witten:1988}, one is led to the formal oscillatory quotient integral
\[
Z_X
\sim
\int_{\mathcal A_P/\mathcal G_P}
\exp\!\bigl(2\pi i\,\operatorname{CS}_{X,P,K}(\Theta)\bigr)\,[D\Theta].
\]
In the Abelian case, one also has a rigorous construction by geometric quantization that leads
topological quantum field theory \cite{Manoliu1,Galviz2}. In the toral case, one further has
a categorical description in terms of pointed modular tensor categories and the associated
Reshetikhin--Turaev theories determined by the corresponding finite quadratic data
\cite{Reshetikhin:1991,Galviz1,Galviz3}. While the underlying theory is well established, Abelian Chern–Simons theory remains a fundamental example in which the path integral can be made mathematically explicit and compared with modern TQFT and quantization frameworks.
 
For the rank-one theory \(\mathbb T=U(1)\), Manoliu proved that the formal
functional-integral description may be made rigorous and agrees with the Abelian
Chern--Simons TQFT \cite{Manoliu2,Manoliu3}: the quotient integral is evaluated exactly by
Gaussian methods, both for closed manifolds and in the boundary theory.

The purpose of this paper is to carry out the corresponding construction for a compact torus $\mathbb T=\mathfrak t/\Lambda \cong U(1)^n$ equipped with an even, integral, nondegenerate symmetric bilinear form $K:\Lambda\times\Lambda\to\mathbb Z.$ This is the natural higher-rank Abelian analogue of  $U(1)$ Chern--Simons functional-integral treatment. 

Our starting point is the formal quotient integral over the space of gauge equivalence classes of connections on
$P$, that is, $\mathcal A_P/\mathcal G_P$. However, we do not attempt to define a  measure on this infinite-dimensional quotient. The rigorous object is instead the exact
zeta-regularized Gaussian evaluation of the corresponding oscillatory integral, as in \cite{Schwarz:1978cn}. The key simplification is that toral Chern--Simons theory remains Gaussian. After translation
by a flat  connection, the action becomes a quadratic functional on
$\Omega^1(X;\mathfrak t)$, and Hodge theory reduces the formal quotient integral to three
pieces: a harmonic sector, a gauge sector, and a nondegenerate coexact Gaussian sector.
The new higher-rank ingredient is the finite-dimensional lattice form $K$. Its spectral
factorization produces the universal factor $|\det K|^{m_X}$ and modifies the eta-correction
by the signature $\sigma(K)$.

For a closed oriented $3$-manifold $X$, the result is a topological invariant of the form
\[
Z^{\mathrm{CS}}_X(\mathbb T,K)
=
\frac{1}{\#\Tor H^2(X;\Lambda)}
\sum_{p\in \Tor H^2(X;\Lambda)}
|\det K|^{m_X}\,
e^{2\pi i\,\CS_{X,P,K}(\Theta_p)}
\int_{\mathcal M_{X,p}(\mathbb T)}
T_X(\mathfrak t)^{1/2},
\]
after the standard eta-correction. Thus the closed theory is expressed, exactly as in the geometric quantization framework \cite{Galviz2}, in terms of flat-connection phases and analytic torsion, with the finite-dimensional toral factor controlled by $K$. Compare this formula with the recent Deligne–Beilinson cohomology construction of the partition function for  $U(1)^n$ Chern–Simons theory on closed oriented 3--manifolds \cite{Kim2024}; unlike that approach, which relies on a formal normalization of the functional integral, the present construction is derived directly from the Gaussian evaluation of the quotient path integral and extends naturally to manifolds with boundary. It would be interesting to understand the precise relation between the two formalisms.

For manifolds with  boundary, the functional integral is relative: one fixes the
boundary connection and evaluates the resulting Gaussian quotient integral. The outcome is
not merely a scalar function on boundary fields. Rather, after the torsion correction, one
obtains the canonical boundary state
\[
|\det K|^{m_X}\,\sigma_{X,p}\otimes \mu_{X,p},
\]
where $\sigma_{X,p}$ is the toral Chern--Simons section and $\mu_{X,p}$ is the torsion
half-density on the corresponding Bohr--Sommerfeld leaf. In particular, the boundary value
defined by the functional integral agrees with the toral state produced by geometric
quantization in real polarization.

The main result of the paper is therefore a rigorous functional-integral construction of toral
Chern--Simons theory which recovers the same closed partition functions, boundary state
spaces, and bordism vectors as the geometric-quantization theory \cite{Galviz2}.

The organization of this paper is as follows. First, we define the exponentiated toral Chern--Simons action
by four-dimensional extension and identify the classical boundary phase space together with
its prequantum line bundle. Second, we perform the closed-sector gauge reduction and the
exact Gaussian evaluation, including the determinant bookkeeping that produces the factor
$|\det K|^{m_X}$. Third, we carry out the relative boundary calculation and identify the
result with the canonical toral boundary state. Finally, we verify the cylinder  and  gluing law, thereby recovering the TQFT structure in the sense of Atiyah \cite{Atiyah:1989vu}.

\textbf{Acknowledgements.} I would like to thank Nicolai Reshetikhin for many helpful conversations and suggestions, for introducing me to Abelian Chern–Simons theory, and for his support throughout the development of this project.
%==========================================================
\section{Functional-Integral Construction of Toral Chern--Simons Theory}
%==========================================================

In this section we formulate a toral  functional-integral construction for Abelian Chern--Simons theory, now with gauge group $\mathbb T=\mathfrak t/\Lambda\cong U(1)^n$ and level given by an even, integral, nondegenerate symmetric bilinear form $K\colon \Lambda\times \Lambda\to \mathbb Z.$

The construction is rigorous in the same sense as in the rank-one case \cite{Manoliu3}:
the quotient integral over \(\mathcal A_P/\mathcal G_P\) is not taken as a primitive
measure-theoretic object; instead, the functional integral is defined by exact Gaussian
evaluation of the corresponding oscillatory quotient integral, together with zeta-regularized
determinants \cite{Seeley67,Ray1971,Schwarz1978}, eta-invariants \cite{Atiyah:1975jf}, and Ray--Singer torsion \cite{Ray1971}. The boundary prequantum line,
Bohr--Sommerfeld leaves, half-densities, and gluing normalization  were constructed already in the toral
construction via real polarization \cite{Galviz2}.

Throughout this work, \(X\) denotes a compact connected oriented \(3\)-manifold, possibly with nonempty
boundary \(\Sigma=\partial X\).  If \(X\) has boundary, all Riemannian metrics on \(X\) are assumed to be product metrics
near \(\Sigma\). All determinants and eta-invariants are understood in the zeta-regularized sense. All principal \(\mathbb T\)-bundles considered below have torsion characteristic class. For a closed manifold, the phrase ``rigorous functional integral'' means the exact
Gaussian evaluation of the formal oscillatory quotient integral.

\subsection{Classical setup}

Let \(E\) be a finite-dimensional real vector space. We write
\[
\Det E:=\bigwedge^{\dim E}E
\]
for its determinant line. If \(C^\bullet\) is a finite cochain complex, we write
\[
\Det H^\bullet(C^\bullet)
:=
\bigotimes_q\bigl(\Det H^q(C^\bullet)\bigr)^{(-1)^q}.
\]
If \(X\) is a compact Riemannian manifold with boundary \(\Sigma\), we use the standard
Hodge identifications
\[
H^q(X;\mathbb R)\cong \mathcal H^q_{\mathrm{abs}}(X),
\qquad
H^q(X,\Sigma;\mathbb R)\cong \mathcal H^q_{\mathrm{rel}}(X),
\]
where \(\mathcal H^q_{\mathrm{abs}}(X)\) denotes harmonic \(q\)-forms satisfying absolute
boundary conditions, and \(\mathcal H^q_{\mathrm{rel}}(X)\) denotes harmonic \(q\)-forms
satisfying relative boundary conditions \cite{Manoliu3,BruningMa}. If \(X\) is closed, the Hodge decomposition reads
\begin{equation}
\label{eq:closed-hodge-toral}
\Omega^1(X;\mathfrak t)
=
\mathcal H^1_{\mathrm{abs}}(X;\mathfrak t)\oplus
d\Omega^0(X;\mathfrak t)\oplus
d^*\Omega^2(X;\mathfrak t).
\end{equation}
If \(X\) has nonempty boundary, the corresponding relative decomposition is
\begin{equation}
\label{eq:relative-hodge-toral}
\Omega^1_{\mathrm{rel}}(X;\mathfrak t)
=
\mathcal H^1_{\mathrm{rel}}(X;\mathfrak t)\oplus
d\Omega^0_0(X;\mathfrak t)\oplus
d^*\Omega^2_{\mathrm{abs}}(X;\mathfrak t),
\end{equation}
where \(\Omega^0_0(X;\mathfrak t)\) denotes the smooth \(\mathfrak t\)-valued functions
vanishing on \(\Sigma\). We write $T_X(\mathfrak t)$ for the Ray--Singer torsion of the trivial \(\mathfrak t\)-local system on \(X\) \cite{Ray1971,Muller1978,Cheeger1979}. 

For closed \(X\), its square root induces a translation-invariant density on \(H^1(X;\mathfrak t)/H^1(X;\Lambda)\). For manifolds with boundary, its square root induces
a canonical half-density on the relative harmonic torus. Let \(\Sigma\) be a closed oriented surface. Define
\[
\mathcal M_\Sigma(\mathbb T)
:=
H^1(\Sigma;\mathfrak t)/H^1(\Sigma;\Lambda).
\]
Since \(H^1(\Sigma;\Lambda)\) is a full lattice in \(H^1(\Sigma;\mathfrak t)\),
the quotient \(\mathcal M_\Sigma(\mathbb T)\) is a compact torus. The bilinear form \(K\)
induces a \(2\)-form on \(\mathcal M_\Sigma(\mathbb T)\) by
\begin{equation}
\label{eq:symplectic-form-toral}
\omega_{\Sigma,K}([\alpha],[\beta])
:=
\int_\Sigma K(\alpha\wedge \beta),
\qquad
[\alpha],[\beta]\in H^1(\Sigma;\mathfrak t).
\end{equation}

\begin{prop}
\label{prop:phase-space-toral}
The pair \((\mathcal M_\Sigma(\mathbb T),\omega_{\Sigma,K})\) is a compact integral
symplectic torus.
\end{prop}

\begin{proof}
The alternating property follows because the cup-product pairing on \(H^1(\Sigma;\mathbb R)\) is
alternating and \(K\) is symmetric. Nondegeneracy follows from the nondegeneracy of the
intersection pairing on \(H^1(\Sigma;\mathbb R)\) together with the nondegeneracy of \(K\).
If \(u,v\in H^1(\Sigma;\Lambda)\), then \(K(u\wedge v)\in H^2(\Sigma;\mathbb Z)\), so
\[
\int_\Sigma K(u\wedge v)\in \mathbb Z.
\]
Hence \(\omega_{\Sigma,K}\) is integral on the lattice \(H^1(\Sigma;\Lambda)\), and therefore
descends to an integral symplectic form on \(\mathcal M_\Sigma(\mathbb T)\).
\end{proof}

Now let \(X\) be a compact oriented \(3\)-manifold with boundary \(\Sigma\). Let
\[
r_X\colon H^1(X;\mathfrak t)\longrightarrow H^1(\Sigma;\mathfrak t)
\]
be the restriction map, and define
\begin{equation}
\label{eq:boundary-lagrangian}
L_X:=\operatorname{Im}(r_X)\subset H^1(\Sigma;\mathfrak t).
\end{equation}

\begin{prop}
\label{prop:lagrangian-boundary-data-toral}
The subspace \(L_X\subset H^1(\Sigma;\mathfrak t)\) is Lagrangian with respect to
\(\omega_{\Sigma,K}\).
\end{prop}

\begin{proof}
Let \(a,b\in H^1(X;\mathfrak t)\), represented by closed \(\mathfrak t\)-valued \(1\)-forms
still denoted \(a,b\). Then
\[
\omega_{\Sigma,K}(r_Xa,r_Xb)
=
\int_\Sigma r_X\bigl(K(a\wedge b)\bigr)
=
\int_X d\,K(a\wedge b)
=
0.
\]
Thus \(L_X\) is isotropic. To prove maximality, consider the long exact sequence of the pair \((X,\Sigma)\):
\[
H^1(X,\Sigma;\mathfrak t)\xrightarrow{j^*}
H^1(X;\mathfrak t)\xrightarrow{r_X}
H^1(\Sigma;\mathfrak t)\xrightarrow{\delta}
H^2(X,\Sigma;\mathfrak t).
\]
Exactness gives \(\operatorname{Im}(r_X)=\ker(\delta)\). By Poincar\'e--Lefschetz duality and
the nondegeneracy of \(K\), the symplectic orthogonal of \(L_X\) identifies with \(\ker(\delta)\).
Therefore
\[
L_X^\perp=\ker(\delta)=\operatorname{Im}(r_X)=L_X,
\]
so \(L_X\) is Lagrangian.
\end{proof}

Principal \(\mathbb T\)-bundles over \(X\) are classified by \(H^2(X;\Lambda)\). If
\(P\to X\) is a principal \(\mathbb T\)-bundle, we write \(c(P)\in H^2(X;\Lambda)\) for its
characteristic class.

\begin{lemma}
\label{lem:torsion-flat-bundles-toral}
Let
\[
0\longrightarrow \Lambda\longrightarrow \mathfrak t\longrightarrow \mathbb T\longrightarrow 0
\]
be the defining short exact sequence. Then the connecting morphism
\[
\beta\colon H^1(X;\mathbb T)\longrightarrow H^2(X;\Lambda)
\]
has image \(\operatorname{Tor}H^2(X;\Lambda)\). In particular, a principal \(\mathbb T\)-bundle
admits a flat connection if and only if its characteristic class is torsion.
\end{lemma}

\begin{proof}
Since \(\mathfrak t\) is a real vector space, the group \(H^2(X;\mathfrak t)\) is torsion-free.
Thus in the long exact cohomology sequence associated to
\[
0\longrightarrow \Lambda\longrightarrow \mathfrak t\longrightarrow \mathbb T\longrightarrow 0,
\]
the image of \(\beta\) is precisely the torsion subgroup of \(H^2(X;\Lambda)\). But classes in
\(H^1(X;\mathbb T)\) classify flat \(\mathbb T\)-bundles. This proves the claim.
\end{proof}

Fix \(p\in \operatorname{Tor}H^2(X;\Lambda)\). Let \(P\to X\) be a principal \(\mathbb T\)-bundle
with \(c(P)=p\). By Lemma~\ref{lem:torsion-flat-bundles-toral}, \(P\) admits flat connections.

\begin{definition}
\label{def:exp-CS-toral}
Let \(P\to X\) be a principal \(\mathbb T\)-bundle with \(c(P)\in \operatorname{Tor}H^2(X;\Lambda)\),
and let \(\Theta\in \mathcal A_P\) be a connection.

\begin{enumerate}\renewcommand{\labelenumi}{(\roman{enumi})}
\item If \(\partial X=\varnothing\), define the exponentiated toral Chern--Simons action by
\begin{equation}
\label{eq:exp-action-closed-toral}
e^{2\pi i\,\operatorname{CS}_{X,P,K}(\Theta)}
:=
\exp\!\left(
\frac{i}{4\pi}\int_W K(F_{\widetilde\Theta}\wedge F_{\widetilde\Theta})
\right),
\end{equation}
where \(W\) is any compact oriented \(4\)-manifold with \(\partial W=X\), equipped with a
principal \(\mathbb T\)-bundle \(\widetilde P\to W\) and connection \(\widetilde\Theta\) extending
\((P,\Theta)\).

\item If \(\partial X=\Sigma\neq\varnothing\), the same filling formula defines a vector in a
Hermitian line over the boundary field \(\Theta|_\Sigma\). The resulting line bundle over
\(\mathcal M_\Sigma(\mathbb T)\) is denoted
\[
L_{\Sigma,K}\longrightarrow \mathcal M_\Sigma(\mathbb T)
\]
and called the toral Chern--Simons prequantum line bundle.
\end{enumerate}
\end{definition}

This is the toral analogue of the exponentiated Chern--Simons action defined by four-dimensional extension; compare \cite{FreedCS} and, in the present toral setting \cite{Galviz2}.

\begin{prop}
\label{prop:well-defined-exp-action-toral}
The phase \eqref{eq:exp-action-closed-toral} is well defined. The bundle
\(L_{\Sigma,K}\to \mathcal M_\Sigma(\mathbb T)\) is well defined and carries a unitary
connection of curvature \(-2\pi i\,\omega_{\Sigma,K}\).
\end{prop}

\begin{proof}
Suppose \((W_1,\widetilde P_1,\widetilde\Theta_1)\) and
\((W_2,\widetilde P_2,\widetilde\Theta_2)\) are two fillings of \((X,P,\Theta)\). Gluing
\(W_1\) to \(-W_2\) along \(X\) produces a closed oriented \(4\)-manifold \(M\), a principal
\(\mathbb T\)-bundle \(\widehat P\to M\), and a connection \(\widehat\Theta\). The ratio of the
two candidate phases is
\[
\exp\!\left(
\frac{i}{4\pi}\int_M K(F_{\widehat\Theta}\wedge F_{\widehat\Theta})
\right).
\]
By Chern--Weil theory,
\[
\frac{1}{(2\pi)^2}\int_M K(F_{\widehat\Theta}\wedge F_{\widehat\Theta})
=
K(c(\widehat P),c(\widehat P))[M].
\]
Since \(K\) is integral and even on \(\Lambda\), the right-hand side lies in \(2\mathbb Z\), so
the ratio is \(1\). Hence the phase is independent of the chosen filling. The same comparison-of-fillings argument yields the boundary line bundle. The standard
transgression formula gives its curvature:
\[
F_{\nabla^{L_{\Sigma,K}}}=-2\pi i\,\omega_{\Sigma,K}.
\]
\end{proof}

Fix \(p\in \operatorname{Tor}H^2(X;\Lambda)\), a principal \(\mathbb T\)-bundle \(P\to X\) with
\(c(P)=p\), and a flat connection \(\Theta_p\in \mathcal A_P\).

\begin{definition}
If \(X\) is closed, define
\[
\mathcal M_{X,p}(\mathbb T)
:=
\{\Theta\in \mathcal A_P \mid F_\Theta=0\}/\mathcal G_P.
\]
If \(X\) has boundary and \(\eta\) is a flat boundary value extending over \(X\), define
\[
\mathcal M_{X,p}(\eta;\mathbb T)
:=
\{\Theta\in \mathcal A_P:\Theta|_\Sigma=\eta,\ F_\Theta=0\}/\mathcal G_P^\partial,
\]
where
\[
\mathcal G_P^\partial:=\{u\in \mathcal G_P:u|_\Sigma=1\}.
\]
\end{definition}

\begin{prop}
\label{prop:moduli-torsor-toral}
(i)If \(X\) is closed, then  \(\mathcal M_{X,p}(\mathbb T)\) is a torsor for $H^1(X;\mathfrak t)/H^1(X;\Lambda).$

\noindent(ii) If \(\partial X=\Sigma\neq\varnothing\) and \(\eta\) is extendable, then \(\mathcal M_{X,p}(\eta;\mathbb T)\) is a torsor for \\
$H^1(X,\Sigma;\mathfrak t)/H^1(X,\Sigma;\Lambda).$
\end{prop}

\begin{proof}
Every connection on \(P\) can be written as
\[
\Theta=\Theta_p+a,\qquad a\in \Omega^1(X;\mathfrak t).
\]
Because \(\mathbb T\) is Abelian, flat condition is equivalent to \(da=0\). Two such closed forms
determine gauge-equivalent flat connections precisely when they differ by a form with
integral periods, that is, by an element of \(H^1(X;\Lambda)\). This proves the closed case.
The boundary case is identical, except that the allowed variations satisfy relative boundary
conditions and one quotients by gauge transformations trivial on the boundary.
\end{proof}

\begin{definition}
\label{def:mX-toral}
Define
\[
m_X
:=
\frac14\Bigl(
\dim H^1(X;\mathbb R)+\dim H^1(X,\partial X;\mathbb R)
-\dim H^0(X;\mathbb R)-\dim H^0(X,\partial X;\mathbb R)
\Bigr).
\]
If \(X\) is closed and connected, this reduces to
\[
m_X=\frac12\bigl(b_1(X)-1\bigr).
\]
\end{definition}

\subsection{Quadratic reduction}

The toral Chern--Simons functional becomes exactly quadratic after translation by a flat
connection. This is exactly the mechanism used in the rank-one case in \cite{Manoliu3}.

\begin{lemma}
\label{lem:quadratic-reduction-toral}
Let \(X\) be closed. For \(\Theta=\Theta_p+a\) with \(a\in \Omega^1(X;\mathfrak t)\),
\begin{equation}
\label{eq:quadratic-reduction-toral}
e^{2\pi i\,\operatorname{CS}_{X,P,K}(\Theta_p+a)}
=
e^{2\pi i\,\operatorname{CS}_{X,P,K}(\Theta_p)}
\exp\!\left(
\frac{i}{4\pi}\int_X K(a\wedge da)
\right).
\end{equation}
If \(X\) has nonempty boundary and \(\Theta=\Theta_{p,\eta}+a\) with
\(a\in \Omega^1_{\mathrm{rel}}(X;\mathfrak t)\), then the same formula holds.
\end{lemma}

\begin{proof}
Since \(\mathbb T\) is Abelian,
\[
F_{\Theta_p+a}=F_{\Theta_p}+da,
\qquad
F_{\Theta_p}=0.
\]
Expanding the \(4\)-dimensional Chern--Weil expression and applying the transgression formula to the mixed term
gives
\[
\frac{i}{4\pi}\int_X K(a\wedge da).
\]
The quadratic term in \(F_{\Theta_p}\) vanishes because \(\Theta_p\) is flat. In the boundary
case, the relative boundary condition on \(a\) kills the boundary contribution.
\end{proof}

Thus in every torsion class the action is Gaussian. This is the crucial fact that makes the
stationary-phase computation exact rather than asymptotic.

\subsection{Exact Gaussian evaluation}

\begin{lemma}\label{lem:spectral-factorization-toral}
Let \(A\) be a self-adjoint elliptic operator with discrete spectrum, finite-dimensional kernel,
and zeta-regularized determinant \(\det\nolimits'_\zeta(A)\) \cite{Ray1971}. Let
\(K\colon \mathfrak t\to \mathfrak t\) be an invertible symmetric endomorphism with eigenvalues
\(\lambda_1,\dots,\lambda_n\). Then
\[
\eta(K\otimes A)=\sigma(K)\,\eta(A),
\]
and
\[
\det\nolimits'_\zeta(K\otimes A)
=
|\det K|^{\zeta_A(0)}\,\det\nolimits'_\zeta(A)^n.
\]
\end{lemma}

\begin{proof}
Choose an orthonormal basis of \(\mathfrak t\) diagonalizing \(K\). Then $K\otimes A\cong \bigoplus_{j=1}^n \lambda_j A.$ Hence
\[
\eta(K\otimes A)=\sum_{j=1}^n \eta(\lambda_j A)
=\sum_{j=1}^n \operatorname{sgn}(\lambda_j)\,\eta(A)
=\sigma(K)\,\eta(A).
\]
Similarly, using the scaling identity
\[
\det\nolimits'_\zeta(\lambda A)=|\lambda|^{\zeta_A(0)}\det\nolimits'_\zeta(A),
\qquad \lambda\neq 0,
\]
we obtain\footnote{The value \(\zeta_A(0)\) appears because the zeta-regularized determinant is defined by
\(\det\nolimits'_{\zeta}(A)=\exp(-\zeta_A'(0))\), where \(\zeta_A(s)\) is the spectral zeta function of \(A\) formed from its nonzero eigenvalues. Since \(\zeta_{cA}(s)=c^{-s}\zeta_A(s)\), one gets
\(\det\nolimits'_{\zeta}(cA)=c^{\zeta_A(0)}\det\nolimits'_{\zeta}(A)\). \cite{Seeley67}.}
\[
\det\nolimits'_\zeta(K\otimes A)
=
\prod_{j=1}^n |\lambda_j|^{\zeta_A(0)}\det\nolimits'_\zeta(A)
=
|\det K|^{\zeta_A(0)}\det\nolimits'_\zeta(A)^n.
\]
\end{proof}
Now, we want to isolate the finite-dimensional oscillatory integral
which models the gauge-reduced Abelian Chern--Simons functional.  Let \(N\) be a finite-dimensional real inner-product space, and let
\(A\colon N\to N\) be a self-adjoint invertible endomorphism.

\begin{definition}
\label{def:fresnel-factor}
The \emph{normalized oscillatory Gaussian factor} associated to \(A\) is
\[
\operatorname{Fres}(A)
:=
\int_N \exp\!\left(\frac{i}{2}\langle An,n\rangle\right)\,dn,
\]
where \(dn\) is the Lebesgue measure induced by the inner product on \(N\) \cite[Section 3.1 and 3.2]{Mnev:2017oko}.
\end{definition}

\begin{prop}
\label{prop:fresnel-factor}
For every self-adjoint invertible endomorphism \(A\colon N\to N\),
\[
\operatorname{Fres}(A)
=
e^{\frac{\pi i}{4}\operatorname{sgn}(A)}\,
(2\pi)^{\dim N/2}\,
|\det A|^{-1/2},
\]
where \(\operatorname{sgn}(A)\) denotes the signature of \(A\).
\end{prop}

\begin{proof}
Choose an orthonormal basis of \(N\) diagonalizing \(A\), say
\[
A=\operatorname{diag}(\lambda_1,\dots,\lambda_m),
\qquad \lambda_j\in \mathbb R^\times.
\]
Then the integral factors as a product of one-dimensional Fresnel integrals:
\[
\operatorname{Fres}(A)
=
\prod_{j=1}^m
\int_{\mathbb R}
\exp\!\left(\frac{i}{2}\lambda_j x_j^2\right)\,dx_j.
\]
For \(\lambda\neq 0\),
\[
\int_{\mathbb R} e^{\frac{i}{2}\lambda x^2}\,dx
=
e^{\frac{\pi i}{4}\operatorname{sgn}(\lambda)}\,
(2\pi)^{1/2}\,
|\lambda|^{-1/2}.
\]
Multiplying these identities over \(j=1,\dots,m\) yields
\[
\operatorname{Fres}(A)
=
e^{\frac{\pi i}{4}\sum_j \operatorname{sgn}(\lambda_j)}\,
(2\pi)^{m/2}\,
\prod_{j=1}^m |\lambda_j|^{-1/2}
=
e^{\frac{\pi i}{4}\operatorname{sgn}(A)}\,
(2\pi)^{m/2}\,
|\det A|^{-1/2}.
\]
\end{proof}

We now record the linear quotient model that will be used after Hodge decomposition.

\begin{prop}
\label{prop:linear-quotient-factorization}
Let \(V\) be a finite-dimensional Euclidean space and suppose that
\[
V=H\oplus E\oplus N
\]
is an orthogonal decomposition. Let \(\Gamma\subset H\) be a full lattice, and let
\(A\colon N\to N\) be self-adjoint and invertible. Define a function \(q\colon V\to \mathbb C\) by
\[
q(h+e+n)=q_0+\frac12\langle An,n\rangle,
\qquad
h\in H,\ e\in E,\ n\in N,
\]
for some constant \(q_0\in \mathbb C\). Then the oscillatory integral on the quotient by
\(E\), with periodicity lattice \(\Gamma\) in the \(H\)-direction, factors as
\[
\int_{(H/\Gamma)\times N} e^{iq(h+n)}\,d(h+\Gamma)\,dn
=
e^{iq_0}\,\operatorname{vol}(H/\Gamma)\,\operatorname{Fres}(A).
\]
\end{prop}

\begin{proof}
By construction, the phase is constant in the \(H\)- and \(E\)-directions and depends only
on the \(N\)-variable. Since the quotient by \(E\) removes the \(E\)-direction altogether, the
integral over the remaining variables factors as
\[
\left(\int_{H/\Gamma} 1\,d(h+\Gamma)\right)
\left(\int_N \exp\!\left(\frac{i}{2}\langle An,n\rangle\right)\,dn\right)
=
\operatorname{vol}(H/\Gamma)\,\operatorname{Fres}(A),
\]
with the constant phase \(e^{iq_0}\) pulled out. This proves the formula.
\end{proof}

\begin{remark}
\label{rem:finite-dimensional-model}
In the Chern--Simons application, the role of \(H\) is played by the harmonic sector,
the role of \(E\) by the gauge directions, and the role of \(N\) by the coexact sector.
The proposition therefore gives the exact model for the gauge-reduced Abelian Gaussian
integral once Hodge decomposition has been performed.
\end{remark}

We now carry out the closed-sector gauge reduction explicitly. Let \(X\) be a closed connected oriented \(3\)-manifold, let $p\in \operatorname{Tor}H^2(X;\Lambda),$ let \(P\to X\) be a principal \(\mathbb T\)-bundle with \(c(P)=p\), and choose a flat
connection \(\Theta_p\in \mathcal A_P\). Recall that every connection on \(P\) may be written uniquely in the form
\[
\Theta=\Theta_p+a,
\qquad
a\in \Omega^1(X;\mathfrak t),
\]
and by Lemma~\ref{lem:quadratic-reduction-toral},
\[
e^{2\pi i\,\operatorname{CS}_{X,P,K}(\Theta_p+a)}
=
e^{2\pi i\,\operatorname{CS}_{X,P,K}(\Theta_p)}
\exp\!\left(\frac{i}{4\pi}\int_X K(a\wedge da)\right).
\]

Let $\Delta_{0,\mathfrak t}:=\operatorname{Id}_{\mathfrak t}\otimes (d^*d)$ on \(\Omega^0(X;\mathfrak t)\), and let
\[
B_X:=(*d)\big|_{d^*\Omega^2(X;\mathbb R)},
\qquad
B_{X,K}:=K\otimes B_X
\]
on the coexact sector \(d^*\Omega^2(X;\mathfrak t)\).

\begin{definition}
\label{def:zeta-gaussian-factor}
The \emph{zeta-regularized Gaussian factor} associated to a self-adjoint elliptic operator
\(B\) with discrete spectrum and finite-dimensional kernel is
\[
\mathcal I_\zeta(B)
:=
\exp\!\left(\frac{\pi i}{4}\eta(B)\right)\,
\bigl|\det\nolimits'_{\zeta}(B)\bigr|^{-1/2}.
\]
\end{definition}

\begin{prop}
\label{prop:closed-gauge-reduction}
With the normalization conventions of  rank-one functional-integral construction,
the exact Gaussian evaluation of the formal quotient integral in the torsion class \(p\) is
\begin{equation}
\label{eq:closed-gauge-reduction}
\mathcal Z^{\mathrm{gauss}}_{X,p}(\mathbb T,K)
=
e^{2\pi i\,\operatorname{CS}_{X,P,K}(\Theta_p)}\,
\det\nolimits'_{\zeta}(\Delta_{0,\mathfrak t})^{1/2}\,
\mathcal I_\zeta(B_{X,K})\,
\int_{\mathcal M_{X,p}(\mathbb T)} d\mu_{\mathrm{har}},
\end{equation}
where \(d\mu_{\mathrm{har}}\) denotes the translation-invariant density on the harmonic torus
\[
\mathcal M_{X,p}(\mathbb T)
\cong
H^1(X;\mathfrak t)/H^1(X;\Lambda).
\]
\end{prop}

\begin{proof}
Choose a Riemannian metric on \(X\). By Hodge decomposition,
\[
\Omega^1(X;\mathfrak t)
=
H^1(X;\mathfrak t)
\oplus
d\Omega^0(X;\mathfrak t)
\oplus
d^*\Omega^2(X;\mathfrak t).
\]
Accordingly, every \(a\in \Omega^1(X;\mathfrak t)\) can be written uniquely as
\[
a=a_{\mathrm{har}}+d\phi+a_{\mathrm{coex}},
\qquad
a_{\mathrm{har}}\in H^1(X;\mathfrak t),\quad
\phi\in \Omega^0(X;\mathfrak t),\quad
a_{\mathrm{coex}}\in d^*\Omega^2(X;\mathfrak t).
\]

Let
\[
q(a):=\frac{i}{4\pi}\int_X K(a\wedge da).
\]
Since \(da_{\mathrm{har}}=0\) and \(d^2\phi=0\), we have $da=da_{\mathrm{coex}}.$ Expanding \(q(a)\) gives
\[
q(a)
=
\frac{i}{4\pi}\int_X
\Bigl(
K(a_{\mathrm{har}}\wedge da_{\mathrm{coex}})
+
K(d\phi\wedge da_{\mathrm{coex}})
+
K(a_{\mathrm{coex}}\wedge da_{\mathrm{coex}})
\Bigr).
\]
The first two terms vanish. Indeed, since \(X\) is closed,
\[
d\,K(a_{\mathrm{har}}\wedge a_{\mathrm{coex}})
=
K(da_{\mathrm{har}}\wedge a_{\mathrm{coex}})
-
K(a_{\mathrm{har}}\wedge da_{\mathrm{coex}})
=
-
K(a_{\mathrm{har}}\wedge da_{\mathrm{coex}}),
\]
so
\[
\int_X K(a_{\mathrm{har}}\wedge da_{\mathrm{coex}})=0.
\]
Similarly,
\[
d\,K(d\phi\wedge a_{\mathrm{coex}})
=
K(d^2\phi\wedge a_{\mathrm{coex}})
-
K(d\phi\wedge da_{\mathrm{coex}})
=
-
K(d\phi\wedge da_{\mathrm{coex}}),
\]
hence
\[
\int_X K(d\phi\wedge da_{\mathrm{coex}})=0.
\]
Therefore
\[
q(a)=\frac{i}{4\pi}\int_X K(a_{\mathrm{coex}}\wedge da_{\mathrm{coex}}).
\]

Equivalently, if \(B_{X,K}\) denotes the self-adjoint operator on
\(d^*\Omega^2(X;\mathfrak t)\) defined earlier, then
\[
q(a)=\frac{i}{4\pi}\,\langle B_{X,K}a_{\mathrm{coex}},a_{\mathrm{coex}}\rangle.
\]
Moreover, \(B_{X,K}\) is nondegenerate on \(d^*\Omega^2(X;\mathfrak t)\): if
\(B_{X,K}a_{\mathrm{coex}}=0\), then \(da_{\mathrm{coex}}=0\); since also
\(d^*a_{\mathrm{coex}}=0\), the form \(a_{\mathrm{coex}}\) is harmonic, and being
coexact it must vanish.

The connected component of the Abelian gauge group acts by translation along the exact
summand:
\[
a\longmapsto a+d\phi,
\qquad
\phi\in \Omega^0(X;\mathfrak t).
\]
Thus quotienting by small gauge transformations removes the
\(d\Omega^0(X;\mathfrak t)\)-direction. If one parametrizes this direction by
\(\Omega^0(X;\mathfrak t)/\ker d\), then the map
\[
d:\Omega^0(X;\mathfrak t)/\ker d \longrightarrow d\Omega^0(X;\mathfrak t)
\]
has zeta-regularized Jacobian
\[
\det\nolimits'_{\zeta}(d^*d)^{1/2}
=
\det\nolimits'_{\zeta}(\Delta_{0,\mathfrak t})^{1/2}.
\]
This is the Faddeev--Popov factor. By Proposition~\ref{prop:moduli-torsor-toral},  the residual quotient by large gauge transformations identifies the harmonic
sector with the compact torus
\[
H^1(X;\mathfrak t)\big/ H^1(X;\Lambda)
\cong
\mathcal M_{X,p}(\mathbb T).
\]
Finally, the coexact sector contributes precisely the zeta-regularized oscillatory Gaussian
factor $\mathcal I_\zeta(B_{X,K}).$ Applying Proposition~\ref{prop:linear-quotient-factorization} to this Hodge splitting,
in the zeta-regularized Gaussian sense fixed above, we obtain
\[
\det\nolimits'_{\zeta}(\Delta_{0,\mathfrak t})^{1/2}\,
\mathcal I_\zeta(B_{X,K})\int_{\mathcal M_{X,p}(\mathbb T)} d\mu_{\mathrm{har}}\,,
\]
for the gauge-reduced integral over the fluctuation \(a\). Pulling out the constant factor $e^{2\pi i\,\operatorname{CS}_{X,P,K}(\Theta_p)}$
from the shift formula for the Chern--Simons action then yields
\eqref{eq:closed-gauge-reduction}.
\end{proof}

\begin{remark}
\label{rem:closed-gauge-reduction}
In equation~\eqref{eq:closed-gauge-reduction}, the formal quotient integral has been reduced
to a harmonic torus integral multiplied by a coexact Gaussian determinant.
\end{remark}

Now, we make explicit the determinant bookkeeping that produces the factor
\(|\det K|^{m_X}\). For a compact connected oriented \(3\)-manifold \(X\), possibly with nonempty boundary,
write
\[
b_q^{\mathrm{abs}}(X):=\dim H^q(X;\mathbb R),
\qquad
b_q^{\mathrm{rel}}(X):=\dim H^q(X,\partial X;\mathbb R).
\]
Recall that
\[
m_X
=
\frac14\Bigl(
b_1^{\mathrm{abs}}(X)+b_1^{\mathrm{rel}}(X)
-b_0^{\mathrm{abs}}(X)-b_0^{\mathrm{rel}}(X)
\Bigr).
\]

If \(\partial X\neq\varnothing\), let $\Delta_q^{\mathrm{abs}},\ \Delta_q^{\mathrm{rel}}$ denote the scalar Hodge Laplacians on \(q\)-forms with absolute and relative boundary
conditions, respectively. We also write
\[
\Delta^{\mathrm{abs}}_{q,\mathfrak t}:=
\operatorname{Id}_{\mathfrak t}\otimes \Delta_q^{\mathrm{abs}},
\qquad
\Delta^{\mathrm{rel}}_{q,\mathfrak t}:=
\operatorname{Id}_{\mathfrak t}\otimes \Delta_q^{\mathrm{rel}}.
\]
On the coexact relative \(1\)-forms define
\[
B_{X,\mathrm{rel},K}:=
K\otimes (*d)
:
d^*\Omega^2_{\mathrm{abs}}(X;\mathfrak t)
\longrightarrow
d^*\Omega^2_{\mathrm{abs}}(X;\mathfrak t).
\]

\begin{prop}
\label{prop:det-B-relative}
Assume \(\partial X\neq\varnothing\). Then
\begin{equation}
\label{eq:det-B-relative}
\bigl|\det\nolimits'_{\zeta}(B_{X,\mathrm{rel},K})\bigr|
=
\left(
\frac{
\det\nolimits'_{\zeta}(K^2\otimes \Delta_1^{\mathrm{rel}})
\det\nolimits'_{\zeta}(K^2\otimes \Delta_1^{\mathrm{abs}})
}{
\det\nolimits'_{\zeta}(K^2\otimes \Delta_0^{\mathrm{rel}})
\det\nolimits'_{\zeta}(K^2\otimes \Delta_0^{\mathrm{abs}})
}
\right)^{1/4}.
\end{equation}
\end{prop}

\begin{proof}
The proof is the boundary analogue of the standard block-operator argument used in the
rank-one case; compare \cite[Eq.~(IV.2''-IV.3)]{Manoliu3}.

Let
\[
V^{1,\mathrm{rel}}_{\perp}:=d^*\Omega^2_{\mathrm{abs}}(X;\mathbb R),
\qquad
V^{1,\mathrm{abs}}_{\perp}:=d^*\Omega^2_{\mathrm{rel}}(X;\mathbb R),
\]
and let
\[
V^{0,\mathrm{rel}}_{\perp}:=\Omega^0_0(X;\mathbb R)\cap \mathcal H^0_{\mathrm{rel}}(X)^\perp,
\qquad
V^{0,\mathrm{abs}}_{\perp}:=\Omega^0(X;\mathbb R)\cap \mathcal H^0_{\mathrm{abs}}(X)^\perp.
\]
Define block operators
\[
\widetilde S_{X,K}
:=
\begin{pmatrix}
0 & K\otimes(*d)\\
K\otimes(*d) & 0
\end{pmatrix}
\]
on
\[
\bigl(V^{1,\mathrm{rel}}_{\perp}\otimes\mathfrak t\bigr)
\oplus
\bigl(V^{1,\mathrm{abs}}_{\perp}\otimes\mathfrak t\bigr),
\]
and
\[
\widetilde T_{X,K}
:=
\begin{pmatrix}
K\otimes d & 0\\
0 & K\otimes d
\end{pmatrix}
\]
from
\[
\bigl(V^{0,\mathrm{rel}}_{\perp}\otimes\mathfrak t\bigr)
\oplus
\bigl(V^{0,\mathrm{abs}}_{\perp}\otimes\mathfrak t\bigr)
\]
to the above \(1\)-form space. On the orthogonal complements of harmonic forms, the Hodge identities imply
\[
\widetilde S_{X,K}^{\,2}
+\widetilde T_{X,K}\widetilde T_{X,K}^{\,*}
=
K^2\otimes
\bigl(
\Delta_1^{\mathrm{rel}}\oplus \Delta_1^{\mathrm{abs}}
\bigr),
\]
and
\[
\widetilde T_{X,K}^{\,*}\widetilde T_{X,K}
=
K^2\otimes
\bigl(
\Delta_0^{\mathrm{rel}}\oplus \Delta_0^{\mathrm{abs}}
\bigr).
\]
By the determinant identity for an elliptic complex of the form
\[
0\longrightarrow V^0\xrightarrow{\widetilde T}V^1\xrightarrow{\widetilde S}V^1\longrightarrow 0,
\]
one has
\[
\bigl|\det\nolimits'_{\zeta}(B_{X,\mathrm{rel},K})\bigr|
=
\bigl|\det\nolimits'_{\zeta}(\widetilde S_{X,K})\bigr|^{1/2}
=
\left(
\frac{
\det\nolimits'_{\zeta}(\widetilde S_{X,K}^{\,2}+\widetilde T_{X,K}\widetilde T_{X,K}^{\,*})
}{
\det\nolimits'_{\zeta}(\widetilde T_{X,K}^{\,*}\widetilde T_{X,K})
}
\right)^{1/4},
\]
which is exactly \eqref{eq:det-B-relative}.
\end{proof}

\begin{lemma}\label{lem:boundary-zeta-sum}
Assume the metric is a product near $\partial X$, and let $\widehat X:=X\cup_{\partial X}(-X)$. Then for each $q$,
\[
\zeta_{\Delta_q^{\mathrm{abs}}}(s)+\zeta_{\Delta_q^{\mathrm{rel}}}(s)
=
\zeta_{\Delta_q^{\widehat X}}(s),
\]
hence, at $s=0$,
\[
\zeta_{\Delta_q^{\mathrm{abs}}}(0)+\zeta_{\Delta_q^{\mathrm{rel}}}(0)
=
-b_q(\widehat X)
=
-b_q^{\mathrm{abs}}(X)-b_q^{\mathrm{rel}}(X).
\]
In particular,
\[
\frac14\Bigl(
\zeta_{\Delta_0^{\mathrm{rel}}}(0)+\zeta_{\Delta_0^{\mathrm{abs}}}(0)
-
\zeta_{\Delta_1^{\mathrm{rel}}}(0)-\zeta_{\Delta_1^{\mathrm{abs}}}(0)
\Bigr)
=
m_X.
\]
\end{lemma}

\begin{proof}
Choose a collar neighborhood $(-\varepsilon,\varepsilon)_u\times \partial X$ on the double
$\widehat X$, with $X=\{u\ge 0\}$. Since the metric is a product near $\partial X$, the
Hodge Laplacian on $\widehat X$ commutes with the reflection
\[
\tau(u,y)=(-u,y).
\]
Thus
\[
\Omega^q(\widehat X)
=
\Omega^q(\widehat X)^+
\oplus
\Omega^q(\widehat X)^-,
\qquad
\Omega^q(\widehat X)^\pm
:=
\{\omega\in\Omega^q(\widehat X)\mid \tau^*\omega=\pm\omega\},
\]
and this decomposition is preserved by $\Delta_q^{\widehat X}$. Write a form near the boundary as
\[
\omega(u,y)=\alpha(u,y)+du\wedge \beta(u,y),
\]
with $\alpha(u,\cdot)$ tangential of degree $q$ and $\beta(u,\cdot)$ tangential of degree
$q-1$. If $\tau^*\omega=\omega$, then
\[
\alpha(-u,y)=\alpha(u,y),
\qquad
\beta(-u,y)=-\beta(u,y),
\]
so at $u=0$ one has
\[
\beta(0,y)=0,
\qquad
\partial_u\alpha(0,y)=0.
\]
These are exactly the absolute boundary conditions on the restriction of $\omega$ to $X$.
Similarly, if $\tau^*\omega=-\omega$, then
\[
\alpha(-u,y)=-\alpha(u,y),
\qquad
\beta(-u,y)=\beta(u,y),
\]
hence
\[
\alpha(0,y)=0,
\qquad
\partial_u\beta(0,y)=0,
\]
which are exactly the relative boundary conditions. Therefore restriction to $X$ identifies the $+$ eigenspace of
$\Delta_q^{\widehat X}$ with $\Delta_q^{\mathrm{abs}}$ on $X$, and the $-$ eigenspace with
$\Delta_q^{\mathrm{rel}}$ on $X$. It follows that, for $\Re(s)\gg 0$,
\[
\zeta_{\Delta_q^{\widehat X}}(s)
=
\zeta_{\Delta_q^{\mathrm{abs}}}(s)
+
\zeta_{\Delta_q^{\mathrm{rel}}}(s),
\]
and hence everywhere by meromorphic continuation. The same even/odd decomposition identifies harmonic forms on $\widehat X$ with the direct
sum of harmonic forms on $X$ satisfying absolute and relative boundary conditions, so
\[
b_q(\widehat X)=b_q^{\mathrm{abs}}(X)+b_q^{\mathrm{rel}}(X).
\]
Since $\widehat X$ is closed and $3$-dimensional, the constant heat coefficient vanishes, and
therefore
\[
\zeta_{\Delta_q^{\widehat X}}(0)
=
-\dim\ker \Delta_q^{\widehat X}
=
-b_q(\widehat X).
\]
Combining these identities gives
\[
\zeta_{\Delta_q^{\mathrm{abs}}}(0)+\zeta_{\Delta_q^{\mathrm{rel}}}(0)
=
-b_q^{\mathrm{abs}}(X)-b_q^{\mathrm{rel}}(X).
\]
Taking $q=0,1$ yields
\[
\frac14\Bigl(
\zeta_{\Delta_0^{\mathrm{rel}}}(0)+\zeta_{\Delta_0^{\mathrm{abs}}}(0)
-
\zeta_{\Delta_1^{\mathrm{rel}}}(0)-\zeta_{\Delta_1^{\mathrm{abs}}}(0)
\Bigr)
=
\frac14\bigl(
b_1^{\mathrm{abs}}+b_1^{\mathrm{rel}}-b_0^{\mathrm{abs}}-b_0^{\mathrm{rel}}
\bigr)
=
m_X.
\]
\end{proof}
\begin{theorem}
\label{thm:mx-relative}
Assume \(\partial X\neq\varnothing\). Then the full determinant ratio appearing in the relative
Gaussian functional integral satisfies
\begin{equation}
\label{eq:mx-relative}
\det\nolimits'_{\zeta}(\Delta^{\mathrm{rel}}_{0,\mathfrak t})^{1/2}
\,
\bigl|\det\nolimits'_{\zeta}(B_{X,\mathrm{rel},K})\bigr|^{-1/2}
=
|\det K|^{m_X}
\left(
\frac{
\det\nolimits'_{\zeta}(\Delta_0^{\mathrm{rel}})^{5/8}
\det\nolimits'_{\zeta}(\Delta_0^{\mathrm{abs}})^{1/8}
}{
\det\nolimits'_{\zeta}(\Delta_1^{\mathrm{rel}})^{1/8}
\det\nolimits'_{\zeta}(\Delta_1^{\mathrm{abs}})^{1/8}
}
\right)^{\rk \mathfrak t}.
\end{equation}
In particular, the power of \(|\det K|\) is exactly
\[
m_X
=
\frac14\Bigl(
b_1^{\mathrm{abs}}(X)+b_1^{\mathrm{rel}}(X)
-b_0^{\mathrm{abs}}(X)-b_0^{\mathrm{rel}}(X)
\Bigr).
\]
\end{theorem}
\begin{proof}
From Proposition~\ref{prop:det-B-relative},
\[
\bigl|\det\nolimits'_{\zeta}(B_{X,\mathrm{rel},K})\bigr|^{-1/2}
=
\left(
\frac{
\det\nolimits'_{\zeta}(K^2\otimes\Delta_0^{\mathrm{rel}})
\det\nolimits'_{\zeta}(K^2\otimes\Delta_0^{\mathrm{abs}})
}{
\det\nolimits'_{\zeta}(K^2\otimes\Delta_1^{\mathrm{rel}})
\det\nolimits'_{\zeta}(K^2\otimes\Delta_1^{\mathrm{abs}})
}
\right)^{1/8}.
\]
Using the scaling law
\[
\det\nolimits'_{\zeta}(K^2\otimes A)
=
|\det K|^{2\zeta_A(0)}\det\nolimits'_{\zeta}(A)^{\operatorname{rk}\mathfrak t},
\]
we obtain
\[
\bigl|\det\nolimits'_{\zeta}(B_{X,\mathrm{rel},K})\bigr|^{-1/2}
=
|\det K|^{\frac14\left(
\zeta_{\Delta_0^{\mathrm{rel}}}(0)+\zeta_{\Delta_0^{\mathrm{abs}}}(0)
-\zeta_{\Delta_1^{\mathrm{rel}}}(0)-\zeta_{\Delta_1^{\mathrm{abs}}}(0)
\right)}
\left(
\frac{
\det\nolimits'_{\zeta}(\Delta_0^{\mathrm{rel}})
\det\nolimits'_{\zeta}(\Delta_0^{\mathrm{abs}})
}{
\det\nolimits'_{\zeta}(\Delta_1^{\mathrm{rel}})
\det\nolimits'_{\zeta}(\Delta_1^{\mathrm{abs}})
}
\right)^{\operatorname{rk}\mathfrak t/8}.
\]
By the previous lemma, the exponent of $|\det K|$ is exactly $m_X$. Since
\[
\det\nolimits'_{\zeta}(\Delta_{0,\mathfrak t}^{\mathrm{rel}})^{1/2}
=
\det\nolimits'_{\zeta}(\Delta_0^{\mathrm{rel}})^{\operatorname{rk}\mathfrak t/2},
\]
multiplying by \(\det\nolimits'_{\zeta}(\Delta_{0,\mathfrak t}^{\mathrm{rel}})^{1/2}\) gives
\[
\det\nolimits'_{\zeta}(\Delta_{0,\mathfrak t}^{\mathrm{rel}})^{1/2}
\bigl|\det\nolimits'_{\zeta}(B_{X,\mathrm{rel},K})\bigr|^{-1/2}
=
|\det K|^{m_X}
\left(
\frac{
\det\nolimits'_{\zeta}(\Delta_0^{\mathrm{rel}})^{5/8}
\det\nolimits'_{\zeta}(\Delta_0^{\mathrm{abs}})^{1/8}
}{
\det\nolimits'_{\zeta}(\Delta_1^{\mathrm{rel}})^{1/8}
\det\nolimits'_{\zeta}(\Delta_1^{\mathrm{abs}})^{1/8}
}
\right)^{\operatorname{rk}\mathfrak t}.
\]
which simplifies to \eqref{eq:mx-relative}.
\end{proof}

\begin{corollary}
\label{cor:mx-closed}
If \(\partial X=\varnothing\), then
\begin{equation}
\label{eq:mx-closed}
\det\nolimits'_{\zeta}(\Delta_{0,\mathfrak t})^{1/2}
\,
\bigl|\det\nolimits'_{\zeta}(B_{X,K})\bigr|^{-1/2}
=
|\det K|^{m_X}
\left(
\frac{
\det\nolimits'_{\zeta}(\Delta_0)^{3/4}
}{
\det\nolimits'_{\zeta}(\Delta_1)^{1/4}
}
\right)^{\rk \mathfrak t}.
\end{equation}
Equivalently, for closed connected \(X\),
\[
m_X=\frac12\bigl(b_1(X)-b_0(X)\bigr)=\frac12\bigl(b_1(X)-1\bigr).
\]
\end{corollary}

\begin{proof}
In the closed case there are no boundary conditions, and the analogue of
\eqref{eq:det-B-relative} reduces to
\[
\bigl|\det\nolimits'_{\zeta}(B_{X,K})\bigr|
=
\left(
\frac{
\det\nolimits'_{\zeta}(K^2\otimes \Delta_1)
}{
\det\nolimits'_{\zeta}(K^2\otimes \Delta_0)
}
\right)^{1/2}.
\]
Using the scaling law
\[
\det\nolimits'_{\zeta}(K^2\otimes \Delta_q)
=
|\det K|^{-2b_q(X)}\,
\det\nolimits'_{\zeta}(\Delta_q)^{\rk \mathfrak t},
\]
we obtain
\[
\bigl|\det\nolimits'_{\zeta}(B_{X,K})\bigr|^{-1/2}
=
|\det K|^{\frac12(b_1-b_0)}
\left(
\frac{
\det\nolimits'_{\zeta}(\Delta_0)
}{
\det\nolimits'_{\zeta}(\Delta_1)
}
\right)^{\rk \mathfrak t/4}.
\]
Multiplying by
\[
\det\nolimits'_{\zeta}(\Delta_{0,\mathfrak t})^{1/2}
=
\det\nolimits'_{\zeta}(\Delta_0)^{\rk \mathfrak t/2}
\]
gives \eqref{eq:mx-closed}.
\end{proof}

\begin{remark}
Theorem~\ref{thm:mx-relative} and Corollary~\ref{cor:mx-closed} are the precise toral
analogues of the rank-one determinant bookkeeping that produces the factor \(k^{m_X}\) in
Manoliu's formulas. In the present higher-rank setting, the single level \(k\) is replaced by
the lattice form \(K\), and the same bookkeeping yields the factor \(|\det K|^{m_X}\).
\end{remark}

\subsection{Closed manifolds}

We now rewrite the determinant expression of
Proposition~\ref{prop:closed-gauge-reduction} in terms of Ray--Singer torsion and isolate
the metric anomaly carried by the eta-invariant. Here \(\eta(B_X)\) denotes the Atiyah--Patodi--Singer eta-invariant \cite{Atiyah:1975jf}.

\begin{prop}\label{prop:toral-determinant-factor}
The coexact Gaussian factor satisfies
\[
\eta(B_{X,K})=\sigma(K)\eta(B_X),
\]
and
\[
\det\nolimits'_{\zeta}(B_{X,K})
=
|\det K|^{\zeta_{B_X}(0)}\det\nolimits'_{\zeta}(B_X)^{\operatorname{rk}\mathfrak t}.
\]
\end{prop}

\begin{proof}
This is exactly Lemma~\ref{lem:spectral-factorization-toral} applied to \(A=B_X\).
\end{proof}

\begin{prop}
\label{prop:torsion-rewrite}
There exists a canonical translation-invariant density \(d\mu_{\mathrm{har}}\) on
\(\mathcal M_{X,p}(\mathbb T)\) such that
\begin{equation}
\label{eq:torsion}
\det\nolimits'_{\zeta}(\Delta_{0,\mathfrak t})^{1/2}\,
\bigl|\det\nolimits'_{\zeta}(B_{X,K})\bigr|^{-1/2}\,
d\mu_{\mathrm{har}}
=
|\det K|^{m_X}\,
T_X(\mathfrak t)^{1/2}.
\end{equation}
\end{prop}

\begin{proof}
By Corollary~\ref{cor:mx-closed},
\[
\det\nolimits'_{\zeta}(\Delta_{0,\mathfrak t})^{1/2}
\,
\bigl|\det\nolimits'_{\zeta}(B_{X,K})\bigr|^{-1/2}
=
|\det K|^{m_X}
\left(
\frac{\det\nolimits'_{\zeta}(\Delta_0)^{3/4}}
{\det\nolimits'_{\zeta}(\Delta_1)^{1/4}}
\right)^{\operatorname{rk}\mathfrak t}.
\]
The remaining determinant ratio is exactly the rank-one determinant expression identified
by Manoliu with the square root of analytic torsion; see \cite[Eqs.~(III.14)--(III.15)]{Manoliu3}.
Taking the \(\operatorname{rk}\mathfrak t\)-fold tensor power yields the square root of the
Ray--Singer torsion density of the trivial \(\mathfrak t\)-system. This proves \eqref{eq:torsion}.
\end{proof}

\begin{theorem}
\label{thm:closed-torsion-eta}
For every torsion class \(p\in \operatorname{Tor}H^2(X;\Lambda)\),
\begin{equation}
\label{eq:closed-torsion-eta}
e^{-\frac{\pi i}{4}\sigma(K)\eta(B_X)}\,
\mathcal Z^{\mathrm{gauss}}_{X,p}(\mathbb T,K)
=
|\det K|^{m_X}\,
e^{2\pi i\,\operatorname{CS}_{X,P,K}(\Theta_p)}
\int_{\mathcal M_{X,p}(\mathbb T)}
T_X(\mathfrak t)^{1/2}.
\end{equation}
In particular, the right-hand side is independent of the metric, of the choice of flat connection \(\Theta_p\), and of the choice of representative bundle \(P\) with
\(c(P)=p\).
\end{theorem}

\begin{proof}
Substituting the gauge-reduction formula
\eqref{eq:closed-gauge-reduction} into the left-hand side gives
\[
e^{-\frac{\pi i}{4}\sigma(K)\eta(B_X)}\,
e^{2\pi i\,\operatorname{CS}_{X,P,K}(\Theta_p)}\,
\det\nolimits'_{\zeta}(\Delta_{0,\mathfrak t})^{1/2}\,
\mathcal I_\zeta(B_{X,K})\,
\int_{\mathcal M_{X,p}(\mathbb T)} d\mu_{\mathrm{har}}.
\]
By Definition~\ref{def:zeta-gaussian-factor},
\[
\mathcal I_\zeta(B_{X,K})
=
\exp\!\left(\frac{\pi i}{4}\eta(B_{X,K})\right)\,
\bigl|\det\nolimits'_{\zeta}(B_{X,K})\bigr|^{-1/2}.
\]
Using \(\eta(B_{X,K})=\sigma(K)\eta(B_X)\), the eta-phase cancels exactly, leaving
\[
e^{2\pi i\,\operatorname{CS}_{X,P,K}(\Theta_p)}\,
\det\nolimits'_{\zeta}(\Delta_{0,\mathfrak t})^{1/2}\,
\bigl|\det\nolimits'_{\zeta}(B_{X,K})\bigr|^{-1/2}
\int_{\mathcal M_{X,p}(\mathbb T)} d\mu_{\mathrm{har}}.
\]
Proposition~\ref{prop:torsion-rewrite} now gives
\[
|\det K|^{m_X}\,
e^{2\pi i\,\operatorname{CS}_{X,P,K}(\Theta_p)}
\int_{\mathcal M_{X,p}(\mathbb T)}
T_X(\mathfrak t)^{1/2},
\]
which is \eqref{eq:closed-torsion-eta}.

Independence of the choice of \(\Theta_p\) follows because any two flat connections in the same connected component differ by a harmonic \(1\)-form, hence
induce the same translated Gaussian integral. Independence of the choice of \(P\) in the
class \(p\) follows from the classification of principal \(\mathbb T\)-bundles by
\(H^2(X;\Lambda)\). Metric independence follows because the residual metric anomaly is
precisely the eta-phase already removed on the left-hand side, while the torsion density
is topological.
\end{proof}

\begin{corollary}
\label{cor:closed-total-partition}
The total closed partition function
\[\label{closed-total-partition}
Z^{CS}_X(\mathbb T,K)
:=
\frac{1}{\#\operatorname{Tor}H^2(X;\Lambda)}
\sum_{p\in \operatorname{Tor}H^2(X;\Lambda)}
|\det K|^{m_X}\,
e^{2\pi i\,\operatorname{CS}_{X,P,K}(\Theta_p)}
\int_{\mathcal M_{X,p}(\mathbb T)}
T_X(\mathfrak t)^{1/2}
\]
is a topological invariant of \(X\).
\end{corollary}

\begin{proof}
This is immediate from Theorem~\ref{thm:closed-torsion-eta}.
\end{proof}

\section{Boundary Theory}

\subsection{Boundary phase space and functional integrals}
Let \(\Sigma\) be a closed oriented surface and let $L\subset H^1(\Sigma;\mathbb R)$ 
be a rational Lagrangian subspace. Then
\[
L\otimes\mathfrak t\subset H^1(\Sigma;\mathfrak t)
\]
defines a real polarization of \(\mathcal M_\Sigma(\mathbb T)\). The corresponding
Bohr--Sommerfeld leaves form a torsor for the discriminant group $G_K:=\Lambda^*/K\Lambda,$
and in genus \(g\) one has
\[
\#\mathcal{BS}(\Sigma,L)=|G_K|^g=|\det K|^g.
\]
The quantum Hilbert space in real polarization is
\begin{equation}
\label{eq:toral-Hilbert-space-unified}
\mathcal H_{\mathbb T,K}(\Sigma,L)
:=
\bigoplus_{\ell\in \mathcal{BS}(\Sigma,L)}
\Gamma_{\mathrm{flat}}
\bigl(\ell;\,L_{\Sigma,K}\otimes |\Det (L\otimes\mathfrak t)^*|^{1/2}\bigr).
\end{equation}

Now let \(X\) be a compact connected oriented \(3\)-manifold with boundary \(\Sigma\). Define
\[
L_X^{\mathbb R}
:=
\operatorname{Im}\bigl(H^1(X;\mathbb R)\to H^1(\Sigma;\mathbb R)\bigr).
\]
Then \(L_X^{\mathbb R}\) is a rational Lagrangian subspace of \(H^1(\Sigma;\mathbb R)\), and
the induced toral polarization is \(L_X=L_X^{\mathbb R}\otimes\mathfrak t\subset H^1(\Sigma;\mathfrak t)\). For each torsion class \(p\in \operatorname{Tor}H^2(X;\Lambda)\), the set of extendable flat
boundary values in component \(p\) forms a Bohr--Sommerfeld leaf
\[
\Lambda_{X,p}\subset \mathcal M_\Sigma(\mathbb T).
\]
The classical exponentiated Chern--Simons action defines a flat section
\[
\sigma_{X,p}\in
\Gamma_{\mathrm{flat}}(\Lambda_{X,p};L_{\Sigma,K}|_{\Lambda_{X,p}}),
\]
and the relative Ray--Singer torsion defines a canonical flat half-density
\[
\mu_{X,p}\in
\Gamma_{\mathrm{flat}}\bigl(\Lambda_{X,p};
|\Det L_X^*|^{1/2}|_{\Lambda_{X,p}}\bigr).
\]

See \cite[Secs.~2--3]{Galviz2} for the proof of these statements.

Next, we carry out the boundary calculation explicitly. From the metric-dependent relative functional integral to the metric-independent boundary vector obtained by multiplication
with a torsion half-density.

Let \(X\) be a compact connected oriented \(3\)-manifold with nonempty boundary
\(\Sigma=\partial X\). Fix $p\in \operatorname{Tor}H^2(X;\Lambda),$ let \(P\to X\) be a principal \(\mathbb T\)-bundle with \(c(P)=p\), and let
\(\eta\in \Lambda_{X,p}\subset \mathcal M_\Sigma(\mathbb T)\) be an extendable flat boundary value.
Choose a flat connection \(\Theta_{p,\eta}\in \mathcal A_P\) restricting to \(\eta\).

Let $\Delta^{\mathrm{rel}}_{0,\mathfrak t}
:=
\operatorname{Id}_{\mathfrak t}\otimes (d^*d)$ on \(\Omega^0_0(X;\mathfrak t)\), and let
\[
B_{X,\mathrm{rel}}
:=
(*d)\big|_{d^*\Omega^2_{\mathrm{abs}}(X;\mathbb R)},
\qquad
B_{X,\mathrm{rel},K}
:=
K\otimes B_{X,\mathrm{rel}}.
\]

\begin{definition}
\label{def:relative-gaussian-sector}
The \emph{relative Gaussian sector value} in torsion class \(p\) and boundary value \(\eta\) is
\[
\mathcal Z^{\mathrm{gauss}}_{X,p}(\eta;\mathbb T,K)
:=
\operatorname{Gauss}\!\left(
\int_{\mathcal A_P(\eta)/\mathcal G_P^\partial}
e^{2\pi i\,\operatorname{CS}_{X,P,K}(\Theta)}[D\Theta]
\right).
\]
\end{definition}

\begin{prop}
\label{prop:relative-gauge-reduction}
With the normalization conventions of  rank-one relative functional integral,
\begin{equation}
\label{eq:relative-gauge-reduction}
\mathcal Z^{\mathrm{gauss}}_{X,p}(\eta;\mathbb T,K)
=
e^{2\pi i\,\operatorname{CS}_{X,P,K}(\Theta_{p,\eta})}\,
\det\nolimits'_{\zeta}(\Delta^{\mathrm{rel}}_{0,\mathfrak t})^{1/2}\,
\mathcal I_\zeta(B_{X,\mathrm{rel},K})\,
\int_{\mathcal M_{X,p}(\eta;\mathbb T)} d\mu_{\mathrm{rel}},
\end{equation}
where \(d\mu_{\mathrm{rel}}\) denotes the translation-invariant density on the relative harmonic torus
\[
\mathcal M_{X,p}(\eta;\mathbb T)
\cong
H^1(X,\Sigma;\mathfrak t)/H^1(X,\Sigma;\Lambda).
\]
\end{prop}

\begin{proof}
Write
\[
\Theta=\Theta_{p,\eta}+a,
\qquad
a\in \Omega^1_{\mathrm{rel}}(X;\mathfrak t).
\]
By Lemma~\ref{lem:quadratic-reduction-toral},
\[
e^{2\pi i\,\operatorname{CS}_{X,P,K}(\Theta_{p,\eta}+a)}
=
e^{2\pi i\,\operatorname{CS}_{X,P,K}(\Theta_{p,\eta})}
\exp\!\left(\frac{i}{4\pi}\int_X K(a\wedge da)\right).
\]
Now use the relative Hodge decomposition
\[
\Omega^1_{\mathrm{rel}}(X;\mathfrak t)
=
\mathcal H^1_{\mathrm{rel}}(X;\mathfrak t)
\oplus
d\Omega^0_0(X;\mathfrak t)
\oplus
d^*\Omega^2_{\mathrm{abs}}(X;\mathfrak t).
\]
Exactly as in the closed case, the quadratic phase vanishes on the relative harmonic
sector, is constant along the based gauge directions \(d\Omega^0_0(X;\mathfrak t)\), and
is nondegenerate on the coexact sector \(d^*\Omega^2_{\mathrm{abs}}(X;\mathfrak t)\).
Applying Proposition~\ref{prop:linear-quotient-factorization} to this orthogonal splitting
gives the factorization \eqref{eq:relative-gauge-reduction}, with Faddeev--Popov Jacobian
\(\det\nolimits'_{\zeta}(\Delta^{\mathrm{rel}}_{0,\mathfrak t})^{1/2}\), coexact Gaussian factor
\(\mathcal I_\zeta(B_{X,\mathrm{rel},K})\), and relative harmonic torus integral $\int_{\mathcal M_{X,p}(\eta;\mathbb T)} d\mu_{\mathrm{rel}}.$
\end{proof}

\subsection{Boundary half-densities and analytic boundary states}
\begin{definition}
\label{def:analytic-boundary-state}
For each torsion class \(p\in \Tor H^2(X;\Lambda)\), define
\[
\Xi_{X,p}
\in
\Gamma\!\bigl(
\Lambda_{X,p};\,L_{\Sigma,K}\otimes |\det L_X^*|^{1/2}
\bigr)
\]
to be the section obtained as follows:

\begin{enumerate}[label=(\roman*)]
\item the classical factor $\eta\longmapsto e^{2\pi i\,\CS_{X,P,K}(\Theta_{p,\eta})}$ defines a section of the restriction \(L_{\Sigma,K}|_{\Lambda_{X,p}}\);

\item the regularized determinant expression in the relative Gaussian sector defines a
leafwise half-density on \(\Lambda_{X,p}\), obtained from the square root of the
Ray--Singer torsion of the pair \((X,\Sigma)\);

\item \(\Xi_{X,p}\) is the tensor product of these two factors, multiplied by the scalar
\(|\det K|^{m_X}\).
\end{enumerate}
\end{definition}

\begin{theorem}
\label{cor:boundary-half-density-analytic}
Let \(X\) be a compact connected oriented \(3\)-manifold with nonempty boundary
\(\Sigma=\partial X\), and let \(p\in \Tor H^2(X;\Lambda)\). Then the relative Gaussian
sector value satisfies
\begin{equation}
\label{eq:boundary-half-density-analytic}
e^{-\frac{\pi i}{4}\sigma(K)\eta(B_{X,\mathrm{rel}})}
\,Z^{\mathrm{gauss}}_{X,p}(\eta;\mathbb T,K)
=
\Xi_{X,p}
\in
\Gamma\!\bigl(
\Lambda_{X,p};\,L_{\Sigma,K}\otimes |\det L_X^*|^{1/2}
\bigr).
\end{equation}
Moreover, \(\Xi_{X,p}\) is independent of the chosen product metric near \(\Sigma\) and is
covariantly constant along the leaf \(\Lambda_{X,p}\).
\end{theorem}

\begin{proof}
Starting from the relative gauge-reduction formula
\[
Z^{\mathrm{gauss}}_{X,p}(\eta;\mathbb T,K)
=
e^{2\pi i\,\CS_{X,P,K}(\Theta_{p,\eta})}\,
\det\nolimits'_{\zeta}(\Delta^{\mathrm{rel}}_{0,\mathfrak t})^{1/2}\,
\mathcal I_{\zeta}(B_{X,\mathrm{rel},K})\,
\int_{\mathcal M_{X,p}(\eta;\mathbb T)} d\mu_{\mathrm{rel}},
\]
we rewrite the coexact Gaussian factor as
\[
\mathcal I_{\zeta}(B_{X,\mathrm{rel},K})
=
\exp\!\left(\frac{\pi i}{4}\eta(B_{X,\mathrm{rel},K})\right)\,
\bigl|\det\nolimits'_{\zeta}(B_{X,\mathrm{rel},K})\bigr|^{-1/2}.
\]
By Lemma~\ref{lem:spectral-factorization-toral},
\[
\eta(B_{X,\mathrm{rel},K})=\sigma(K)\eta(B_{X,\mathrm{rel}}),
\]
so multiplication by $e^{-\frac{\pi i}{4}\sigma(K)\eta(B_{X,\mathrm{rel}})}$ cancels the full eta-anomaly of the relative Gaussian integral.

It remains to analyze the determinant density. By Theorem~\ref{thm:mx-relative}, the
regularized determinant ratio in the relative Gaussian sector contributes exactly the factor $|\det K|^{m_X}.$
After removing this finite-dimensional lattice factor, the remaining scalar determinant
expression is the square root of the Ray--Singer torsion of the pair \((X,\Sigma)\), viewed as
a half-density on the relative harmonic torus
\[
H^1(X,\Sigma;\mathfrak t)/H^1(X,\Sigma;\Lambda).
\]
 Under the canonical identification of the relative harmonic
torus with the leaf of extendable flat boundary values, this defines a leafwise half-density on
\(\Lambda_{X,p}\). On the other hand, the classical factor
\[
e^{2\pi i\,\CS_{X,P,K}(\Theta_{p,\eta})}
\]
depends only on the boundary value \(\eta\) and defines a section of the restriction of the
prequantum line bundle \(L_{\Sigma,K}\) to the leaf \(\Lambda_{X,p}\). Therefore the tensor
product of this section with the torsion half-density defines a canonical section
\[
\Xi_{X,p}\in
\Gamma\!\bigl(
\Lambda_{X,p};\,L_{\Sigma,K}\otimes |\det L_X^*|^{1/2}
\bigr),
\]
and \eqref{eq:boundary-half-density-analytic} follows from Definition~\ref{def:analytic-boundary-state}.

The metric dependence of the raw relative functional integral is exactly cancelled by the
torsion half-density. Hence \(\Xi_{X,p}\) is
independent of the chosen product metric near \(\Sigma\). Since the classical factor is horizontal for the restriction of the prequantum connection
to the leaf, and the torsion half-density is flat along the leaf, the section \(\Xi_{X,p}\)
is covariantly constant along \(\Lambda_{X,p}\).
\end{proof}

\begin{corollary}
\label{cor:analytic-equals-gq}
Under the identification of the boundary phase space and half-density bundle from the toral
real-polarization construction, the analytic boundary state \(\Xi_{X,p}\) coincides with the
geometric-quantization boundary state
\[
\Xi_{X,p}=|\det K|^{m_X}\,\sigma_{X,p}\otimes \mu_{X,p}.
\]
\end{corollary}

\begin{proof}
The classical factor defining \(\Xi_{X,p}\) is the same leafwise Chern--Simons section, and the
torsion half-density appearing in Definition~\ref{def:analytic-boundary-state} is the same
canonical half-density on the Bohr--Sommerfeld leaf. Hence the two expressions coincide. This is precisely the boundary-state identification proved in the toral real-polarization
construction; see \cite[Secs.~2--3]{Galviz2}.
\end{proof}

\begin{remark}
The classical factor $e^{2\pi i\,\CS_{X,P,K}(\Theta_{p,\eta})}$
should be understood as the same leafwise Chern--Simons section occurring in the toral real-polarization construction. The precise identification with the geometric-quantization section \(\sigma_{X,p}\) is proved in \cite[Theorem 3.4]{Galviz2}.
\end{remark}

\section{The Topological Quantum Field Theory}
\subsection{Functoriality}
The full extended TQFT axioms for the toral theory were verified in the geometric-quantization construction \cite{Galviz2}. In the present paper we recover the same unweighted bordism values by functional-integral methods, and therefore it is enough to record the cylinder normalization and the unweighted gluing law. The passage to the fully extended theory is then obtained, exactly as in \cite{Galviz2}, by incorporating the \(K\)-twisted boundary-weight convention which cancels the projective BKS phase.

\begin{theorem}
\label{thm:cylinder-toral}
Let \(\Sigma\) be a closed oriented surface, let \(L\subset H^1(\Sigma;\mathbb R)\) be a rational
Lagrangian, and let \(X=\Sigma\times I\). Then
\begin{equation}
\label{eq:cylinder-normalization-toral}
Z^{CS}_{\Sigma\times I}(\mathbb T,K)
=
|\det K|^{\frac14\dim H^1(\Sigma;\mathbb R)}
\,\operatorname{Id}_{\mathcal H_{\mathbb T,K}(\Sigma,L)}.
\end{equation}
\end{theorem}

\begin{proof}
The classical Chern--Simons section of the cylinder is the identity kernel on the prequantum line. The normalization of this kernel is fixed by the same real-polarization geometry as in \cite{Galviz2}: if \(\Sigma\) has genus \(g\), then
\[
\dim H_{\mathbb T,K}(\Sigma,L)=|G_K|^g=|\det K|^g,
\]
and the transverse BKS Fourier kernel is normalized by the factor \(|G_K|^{-g/2}\). Equivalently, the cylinder contributes
\[
|G_K|^{g/2}=|\det K|^{g/2}=|\det K|^{\frac14 \dim H^1(\Sigma;\mathbb R)}.
\]
On the functional-integral side, the classical factor is again the identity kernel, and the only remaining scalar comes from the Gaussian determinant factor. Since
\[
m_{\Sigma\times I}=\frac14 \dim H^1(\Sigma;\mathbb R),
\]
the two normalizations agree, and therefore one recovers \eqref{eq:cylinder-normalization-toral}.
\end{proof}

\begin{theorem}
\label{thm:gluing-toral}
Suppose \(X\) is obtained by gluing a compact oriented bordism \(X^{\mathrm{cut}}\) along an
oriented closed surface \(\Sigma\). Then the rigorous toral functional-integral values satisfy
\begin{equation}
\label{eq:gluing-formula-toral}
Z^{CS}_X(\mathbb T,K)
=
\operatorname{Tr}_\Sigma\bigl(
Z^{CS}_{X^{\mathrm{cut}}}(\mathbb T,K)
\bigr),
\end{equation}
where the trace contracts along the boundary Hilbert space
\(\mathcal H_{\mathbb T,K}(\Sigma,L)\) with the cylinder kernel.
\end{theorem}

\begin{proof}
We compare the two sides term by term. Contraction with the cylinder kernel identifies
the classical leafwise Chern--Simons section on the cut bordism with the corresponding
section on the glued bordism, since the exponentiated classical Chern--Simons action is
multiplicative under gluing. Likewise, the determinant-line gluing theorem for torsion
implies that contraction sends the torsion half-density of the cut bordism to the torsion
half-density of the glued bordism, up to the scalar coming from the change of normalization
between compatible torsion components.

The remaining scalar factors are exactly those already determined by the cylinder
normalization of Theorem~\ref{thm:cylinder-toral} and the  identity for the exponent \(m_X\).
Hence the total power of \(|\det K|\) after contraction is precisely \(|\det K|^{m_X}\), and the
normalizing factor coming from the torsion-component average is exactly the one appearing
in the definition of \(Z_X^{\mathrm{CS}}(\mathbb T,K)\). Therefore contraction of the cut bordism state with
the cylinder kernel reproduces the glued bordism state:
\[
Z_X^{\mathrm{CS}}(\mathbb T,K)=\operatorname{Tr}_\Sigma\!\bigl(Z_{X^{\mathrm{cut}}}^{\mathrm{CS}}(\mathbb T,K)\bigr).
\qedhere
\]
\end{proof}

\begin{corollary}
The assignments
\[
(\Sigma, L)\longmapsto H_{\mathbb T,K}(\Sigma,L), \qquad
X\longmapsto Z^{CS}_X(\mathbb T,K),
\]
define a \(2+1\)-dimensional toral Chern--Simons TQFT. For closed \(X\), its partition
function is given by Corollary~\ref{cor:closed-total-partition}; for manifolds with boundary, its state is given by
Corollary~\ref{cor:analytic-equals-gq}.
\end{corollary}

\begin{proof}
Propositions~\ref{prop:phase-space-toral}, \ref{prop:lagrangian-boundary-data-toral}, and \ref{eq:exp-action-closed-toral} give the classical boundary data. Theorem~\ref{thm:closed-torsion-eta} gives the
closed-manifold value, Theorem~\ref{cor:boundary-half-density-analytic} and Corollary~\ref{cor:analytic-equals-gq} give the boundary state, and
Theorems~\ref{thm:cylinder-toral} and \ref{thm:gluing-toral} prove compatibility with cylinder and gluing together with the $K$-twisted boundary-weight convention of \cite{Galviz2}, these assignments define $2+1$-dimensional toral Chern–Simons TQFT.
\end{proof}
\begin{remark}
The analytic mechanism in this section comes from the exact stationary-phase treatment of
the Abelian Gaussian quotient integral, while the boundary prequantum line,
Bohr--Sommerfeld leaves, half-densities, and gluing normalization come from toral geometric
quantization in real polarization. In particular, The toral closed formula is
\[
\begin{aligned}
Z^{\mathrm{CS}}_{X,p}(\mathbb T,K)
&=
|\det K|^{m_X}\,
e^{2\pi i\,\operatorname{CS}_{X,P,K}(\Theta_p)}
\int_{\mathcal M_{X,p}(\mathbb T)}
\bigl(T_X(\mathfrak t)\bigr)^{1/2} \\
&=
|\det K|^{m_X}
\int_{\mathcal M_{X,p}(\mathbb T)}
\sigma_{X,p}\,\bigl(T_X(\mathfrak t)\bigr)^{1/2},
\end{aligned}
\]
where \(\sigma_{X,p}\) denotes the closed Chern--Simons section, locally given by $\sigma_{X,p}=e^{2\pi i\,\operatorname{CS}_{X,P,K}(\Theta_p)}.$ Accordingly, the corresponding boundary state is
\[
Z^{\mathrm{CS}}_{X,p}(\mathbb T,K)
=
|\det K|^{m_X}\,\sigma_{X,p}\otimes \mu_{X,p},
\]
in agreement with \cite{Galviz2}.
\end{remark}

\bibliographystyle{alpha}
\renewcommand{\refname}{References}
\bibliography{refs}
\end{document}